\documentclass{llncs}
\usepackage{amssymb}
\usepackage{amsmath}
\usepackage{color}

\newcommand{\cM}{\mathcal M}

\newcommand{\F}{{\mathbb F}}
\newcommand{\fp}{{\mathbb F}_{p}}

\newcommand{\fpn}{{\mathbb F}_{p^n}}
\newcommand{\fpk}{{\mathbb F}_{p^k}}

\newcommand{\ftwo}{{\mathbb F}_2}

\newcommand{\fthree}{{\mathbb F}_{3}}
\newcommand{\fthreen}{{\mathbb F}_{3^n}}

\newcommand{\fthreek}{{\mathbb F}_{3^k}}
\newcommand{\Tr}{{\operatorname{Tr}}}

\begin{document}

\pagestyle{plain}

\title{Ternary Binomial and Trinomial (Vectorial) Bent Functions in the Completed Class $\cM^{\dagger}$}
\titlerunning{Ternary Binomial Bent Functions}

\renewcommand{\thefootnote}{$\dagger$}
\footnotetext{This work was supported by the Research Council of Norway.}

\author{Tor Helleseth \and Alexander Kholosha \and Niki Spithaki}

\institute{The Selmer Center\\
Department of Informatics, University of Bergen\\
P.O. Box 7800, N-5020 Bergen, Norway\\
\email{\symbol{123}tor.helleseth,oleksandr.kholosha,niki.spithaki\symbol{125}@uib.no}}

\maketitle

\thispagestyle{plain}

\begin{abstract}
Two classes of ternary (vectorial) bent functions of degree four with two and three terms in the univariate representation that belong to 
the completed Maiorana-McFarland original class are found. Binomials are mappings $\F_{3^{4k}}\mapsto\fthreek$ given by 
$\Tr^{4k}_k\big(a_1 x^{2(3^k+1)}+a_2 x^{(3^k+1)^2}\big)$, where $a_1$ is a nonsquare in $\F_{3^{4k}}$ and $a_2$ is defined explicitly by 
$a_1$. Particular subclass of such binomial vectorial bent function is represented by the exceptional polynomials over $\fthreek$. Bent 
trinomials are mappings $\F_{3^{2k}}\mapsto\fthree$ given by $\Tr_{2k}\big(a_1 x^{2(3^k+2)} + a_2 x^{3^k+5} + a_3 x^2\big)$ with $k>1$ 
not divisible by four and coefficients explicitly defined depending on the parity of $k$. The proof is based on new criterion that allows 
checking bentness by analyzing the first-order derivatives of $f$ in the direction of a chosen $n/2$-dimensional subspace ($n$ is the 
degree of the domain of $f$ over $\fp$). 
\end{abstract}

\section{Introduction}
 \label{sec:intro}
Boolean bent functions were first introduced by Rothaus in 1976 as an interesting combinatorial object. They are functions from $\ftwo^n$ 
to $\ftwo$ for $n$ even, whose Hamming distance to the set of all affine functions is maximum. Equivalently, absolute square of their Walsh 
transform coefficients is constant. They are characterized by the property that their derivatives in nonzero directions are balanced 
(functions having this property are called {\em perfect nonlinear}). Later, the research in this area was stimulated by the significant 
relation to the topics in mathematics and computer science. Kumar, Scholtz and Welch in \cite{KuScWe85} generalized the notion of Boolean 
bent functions to the case of functions from $(\mathbb Z/q\mathbb Z)^n$ to $\mathbb Z/q\mathbb Z$, where $q$ is the power of a prime $p$ 
(which covers the case of functions from $\fp^n$ to $\fp$, called {\em $p$-ary functions}, which are the subject of the present paper). 
Complete classification of bent functions looks hopeless even in the binary case and the case of $p$-ary bent functions is still more 
complicated. However, many explicit methods are known for constructing bent functions, either from scratch (with so-called primary 
constructions) or from other bent functions (in so-called secondary constructions). 

Take a prime $p$. Functions $f:\fp^n\mapsto\fp$ are called $p$-ary functions; if $p=2$ they are referred to as Boolean functions. Functions 
$F:\fp^n\mapsto\fp^k$ with $k\geq 2$ are called $p$-ary vectorial functions. For a fixed nonzero $\alpha\in\fp^k$, the $p$-ary function 
$f_{\alpha}(x)=\alpha\cdot F(x)$ is called a \emph{component function} of $F$ (here '$\cdot$' is an inner product). Every $p$-ary function 
has a unique multivariate representation called the Algebraic Normal Form (ANF) as a polynomial in $x_1,\ldots,x_n$ over $\fp$, where the 
variables $x_i$ occur with exponents smaller than $p$. The degree of the ANF of $f$ is called the algebraic degree of $f$. Functions of 
algebraic degree $2$ are called quadratic, of algebraic degree $3$ cubic, and of algebraic degree $4$ quartic. The algebraic degree of a 
vectorial function is the largest algebraic degree of its component functions. 

Given $n$ and $p$, the vector space $\fp^n$ can be endowed with the structure of the field $\fpn$ (which allows to define polynomial 
functions). The finite field $\fpk$ is a subfield of $\fpn$ if and only if $k$ divides $n$. The trace mapping from $\fpn$ to the subfield 
$\fpk$ is defined by $\Tr_k^n(x)=\sum_{i=0}^{n/k-1}x^{p^{ik}}$. In the case when $k=1$, we use the notation $\Tr_n(x)$ instead of 
$\Tr_1^n(x)$.

Given a function $f$ mapping $\fpn$ to $\fp$, the {\em Walsh transforms} of $f$ and its inverse are respectively defined by the following 
identities 
\begin{equation}
 \label{eq:W_t}
W_f(y)=\sum_{x\in\fpn}\omega^{f(x)-\Tr_n(xy)}; y\in\fpn\;\mbox{and}\; \omega^{f(x)}=\frac{1}{p^n}\sum_{y\in\fpn}W_f(y)\omega^{\Tr_n(xy)}
\end{equation}
where $\omega$ is the complex primitive $p^{\rm{th}}$ root of unity equal to $e^{\frac{2\pi i}{p}}$, and where the elements of $\fp$ are 
considered as integers modulo $p$. 

As defined in \cite{KuScWe85}, $f$ is a {\em $p$-ary bent function} if all its Walsh transform coefficients satisfy $|W_f(y)|^2=p^n$. A 
$p$-ary function $f$ is bent if and only if all the derivatives $D_af(x)=f(x+a)-f(x)$ in nonzero directions $a$ are balanced, that is, take 
any value in $\fp$ equally often (we say that $f$ is {\em perfect nonlinear}). A $p$-ary bent function $f(x)$ is called {\em weakly 
regular} if we can fix a complex $u$ (called \emph{sign}) having unit magnitude such that $p^{-n/2}W_f(y)=u\omega^{f^*(y)}$ for every 
$y\in\fpn$ (then necessarily $u=\pm 1$ or $u=\pm i$) and $f^*:\fpn\mapsto \fp$ is called the {\it dual} of $f$. If $u=1$, then $f$ is {\em 
regular} (here necessarily $n$ is even or $p\equiv 1\pmod 4$). Vectorial function is called bent if all its component functions are bent. 

Weakly regular bent functions were shown to be useful for constructing certain combinatorial objects such as partial difference sets, 
strongly regular graphs and association schemes (see \cite{TaPoFe10,PoTaFeLi11}). This justifies why the classes of (weakly) regular bent 
functions are of independent interest. For a comprehensive reference on monomial and quadratic $p$-ary bent functions we refer reader to 
\cite{HeKh06_1} and to the general survey on $p$-ary and vectorial bent functions \cite{Mei22}. 

\begin{definition} 
Vectorial functions (resp. $p$-ary functions) $F,G:\fp^n\mapsto\fp^m$ (resp. $f,g:\fp^n\mapsto\fp$) are {\em affine equivalent} if there 
exists an affine automorphism $L$ of $\fp^n$ and an affine automorphism $L'$ of $\fp^m$, such that $G=L'\circ F\circ L$ (resp. $g=f\circ 
L$). They are called {\em extended-affine equivalent (in brief, EA-equivalent)} if there exists an affine function $L'':\fp^n\mapsto\fp^m$ 
such that $F+L''$ and $G$ are affine equivalent. A class of functions is complete if it is a union of EA-equivalence classes. The completed 
class is the smallest possible complete class that contains the original one. 
\end{definition}

The families of regular, weakly regular, and non weakly regular bent functions are invariant under EA-equivalence (see \cite{CeMeiPo13_1}). 
The family of vectorial bent functions is also invariant under EA-equivalence. 

It is interesting to find (vectorial) bent functions allowing few terms in their relative trace form. Quadratic bent functions are well 
understood since every such function is (weakly) regular and is EA-equivalent to either $x_1^2+\cdots+x_n^2$ or $dx_1^2+\cdots+x_n^2$ where 
$d$ is a nonsquare in $\fp$ (see \cite[Theorem~3]{Mei22}). From non-quadratic, the most common are regular bent functions with the trace 
form consisting of (so called, Dillon type) exponents divisible by $p^{n/2}-1$. Such bent functions are likely to exist for any $p$ and 
with any number of terms (see, e.g, \cite{QiTaHu18,TaQiHu19,TaXuQiZh21}) and belong to the class of the partial spread functions (see 
\cite{LiLu14}). Exceptional are bent Dillon monomials that exist only for $p=2,3$ (see \cite{HeKh06_1,KoRiVa10}). In addition to quadratic 
and bent Dillon monomials, the only known are ternary Coulter-Matthews perfect nonlinear vectorial functions $x^{(3^k+1)/2}$ on $\fthreen$ 
with $k$ odd and $\gcd(k,n)=1$ (see \cite{CoMa97}) and also ternary vectorial functions with $n=2k$ and $k$ odd given by 
$\Tr^n_k\Big(\xi^{\frac{3^k+1}{4}}x^{\frac{3^n-1}{4}+3^k+1}\Big)$, where $\xi$ is a primitive element of $\fthreen$ (see \cite{HeKh06_1}), 
both of the latter vectorial bent functions consist of weakly regular bent components. Moreover, it is conjectured that the currently known 
list of $p$-ary bent monomials, including one sporadic case, is complete. 

In addition to quadratic and Dillon cases, the only known class of bent binomials is $\Tr^n_{2k}\big(x^{p^{3k}+p^{2k}-p^k+1}+x^2\big)$ with 
$n=4k$, proven in \cite{HeKh10_4} to be a weakly regular for any odd $p$. Also, for some Coulter-Matthews bent functions, the dual is 
trinomial (see \cite{HuZhSh17}). In this paper, we find two new classes of ternary binomials and trinomials having the form 
$\Tr^n_k\left(a_1 x^{d_1}+a_2 x^{d_2}+a_3 x^{d_3}\right)$ with $a_i\in\fthreen$ and particular integer exponents $d_i$ for $i=\{1,2,3\}$. 
These bent functions belong to the completed Maiorana-McFarland (MMF) original class and we apply new technique to prove bentness in this 
case. 

The paper is organized as follows. In Section~\ref{sec:MM_crit}, we prove new criterion for a function to belong to the completed class 
$\cM$ ($\cM$ denotes the set of bent functions in the MMF original class). This result is interesting on its own since it provides a method 
for proving bentness that may work in the case when all the other approaches seem to fail. We use this criterion in the sections that 
follow. We also obtain an interesting connection of functions in the completed class $\cM$ and the hull of a linear code. 

In Section~\ref{sec:MM_binom2}, we present ternary binomial vectorial bent function for $n=4k$ given by $\Tr^n_k\big(a_1 x^{2(3^k+1)}+a_2 
x^{(3^k+1)^2}\big)$, where $a_1$ is a nonsquare in $\fthreen$ and $a_2$ is defined explicitly by $a_1$. In Section~\ref{sec:MM_trinom}, we 
found bent trinomials $\Tr_n\big(a_1 x^{2(3^k+2)} + a_2 x^{3^k+5} + a_3 x^2\big)$ with $n=2k$, $k>1$ not divisible by four, and 
coefficients explicitly defined. Also, vectorial trinomials $\Tr^n_k\big(a(x^{2(3^{k+1}+1)}+x^{3^{k+1}+5}+x^8)$ with $n=2k$, $k>1$ not 
divisible by six, and $a=\xi^{(3^k+1)/2}$ with $\xi$ a primitive element of $\fthreen$, are proven to be bent. All the presented 
(vectorial) functions consist of components that are all regular bent functions from the completed class $\cM$. Binomial functions 
considered in Section~\ref{sec:MM_binom1} are also proven to be bent using the $4$-variable representation and calculating the Walsh 
transform. This approach can work in general and just in this particular case leads to the function in the completed class $\cM$. For small 
dimensions, we checked computationally that these functions are EA-equivalent to the functions in a more general class from 
Section~\ref{sec:MM_binom2} but proving this in general remains open. In the final Section~\ref{sec:MM_excep}, we discuss the subclass of 
the earlier found binomial bent functions that can be represented by exceptional polynomials. 

\section{Proving Bentness in the MMF Original Case}
 \label{sec:MM_crit}
In this section, we introduce Maiorana-McFarland functions and discuss how to prove bentness of a function that belongs to the completed 
class $\cM$. Normally, one would first check that a function is bent and then apply a criterion of being in the completed class $\cM$. Our 
result in Theorem~\ref{th:MM_crit2} provide a new direct method of proving bentness when a function belongs to the completed class $\cM$. 

A well-known primary construction of Boolean functions is that of Maiorana-McFarland (see, e.g., \cite[Definition~46]{Ca96}). The original 
version of this construction (from \cite{Di74}) was for designing bent functions only and represented the input as a pair of vectors of the 
same length. Later was identified a general structure of some Boolean functions where the input is a pair of vectors with possibly 
different lengths. The truth table of these functions is the concatenation of the truth tables of affine functions. This extension allowed 
to construct bent functions but also other types of cryptographically interesting functions. We recall now the $p$-ary generalization of 
this class and specify what is the generalization of the original class. Commonly, MMF functions are represented in bivariate form. 

\begin{definition}
 \label{de:MMF}
Let $n$ and $k$ be any positive integers such that $k<n$. We call Maiorana-McFarland (MMF) function any $n$-variable $p$-ary function of 
the form 
\begin{equation}
 \label{eq:MMFform}
f(x,y)=x\cdot\psi(y)+h(y)\enspace, 
\end{equation}
where $\psi:\fp^{n-k}\mapsto\fp^k$, $h:\fp^{n-k}\mapsto\fp$, and '$\cdot$' is an inner product. The subclass with $n$ even and $k=n/2$ is 
called the MMF original class (see \cite[Definition~46]{Ca20}). 
\end{definition}

Note that it is equivalent to say that $f$ is affine equivalent to a function of the form (\ref{eq:MMFform}) and to say it is EA equivalent 
to a function of this form. Characterization of the completed MMF class of Boolean functions is given in \cite[Proposition~1.3]{PoPo20} 
(and earlier in \cite[p.~102]{Di74} for the MMF original class and \cite[Proposition~54]{Ca20}). In the $p$-ary case, the proof is similar. 

\begin{proposition}
 \label{pr:MM_crit1}
Let $f$ be a $p$-ary function on $\fp^n$. The following statements are equivalent:
\renewcommand{\theenumi}{\roman{enumi}}
\renewcommand{\labelenumi}{(\theenumi)}
\begin{enumerate}
\item Function $f$ belongs to the completed MMF class built from (\ref{eq:MMFform}).
\item There exists a $k$-dimensional subspace $V$ of $\fp^n$ such that the second-order derivatives of $f$ in the direction of elements 
    of $V$ vanish, i.e., $D_{a,b}f(x)\equiv 0$ for any $a,b\in V$.  
\item\label{it:1} There exists a $k$-dimensional subspace $V$ of $\fp^n$ such that $f$ is affine on every coset of $V$.
\end{enumerate}
We say then that $f$ belongs to the completed class of MMF functions {\em with associated ($k$-dimensional) subspace $V$}.
\end{proposition} 

Affine in Item~(\ref{it:1}) means that for any $c\in\fp^n$ and $v\in V$ we have $f(c+v)=v\cdot\psi(c)+h(c)$ for some functions 
$\psi:\fp^n\mapsto\fp^n$ and $h:\fp^n\mapsto\fp$. 

It is well known that an MMF function in the original class is (regular) bent if and only if $\psi$ is a permutation. Denote the set of 
(regular) bent functions in the MMF original class as $\cM$. The MMF original construction of bent functions is essentially vectorial since 
for any $H:\fpk\mapsto\fpk$, function $F:\fpk\times\fpk\mapsto\fpk$ with $F(x,y)=x\psi(y)+H(y)$ is a vectorial bent function if and only if 
$\psi$ is a permutation of $\fpk$. All component functions of $F$ belong to $\cM$. The completed class $\cM$ contains all quadratic bent 
functions. 

The following theorem provides a criterion for a function belong to the completed class $\cM$. The criterion combines the equivalent 
definition of a completed MMF class (see Proposition~\ref{pr:MM_crit1}) and a weaker condition just for some of the first-order derivatives 
to be balanced, that is sufficient to conclude that the function is bent. To apply the criterion, we need to analyse the first-order 
derivatives of a function in the direction of elements taken from the associated subspace. As illustrated in Sections~\ref{sec:MM_binom2} 
and \ref{sec:MM_trinom}, this criterion can be used in the cases when other methods look hopeless. This may be rather technical but still 
doable while proving EA equivalence or calculating exponential sums directly seems impossible. 

\begin{theorem}
 \label{th:MM_crit2}
Take $n=2m$ and function $f:\fp^n\mapsto\fp$. Then $f$ is a (regular) bent function that belongs to the completed class $\cM$ if and only 
if there exists an $m$-dimensional subspace $V$ of $\fp^n$ such that $f$ is affine on every coset of $V$ (then $f(w+v)=v\cdot\pi(w)+g(w)$ 
for some $\pi:W\mapsto\fp^n$, $g:W\mapsto\fp$, where $(v,w)\in V\times W$ and $W$ is a subspace supplementary to $V$) and $W$ can be chosen 
such that $\pi(W)$ is a subspace supplementary to $V^{\bot}$. In particular, the latter condition on $\pi(W)$ holds in the following two 
extreme cases: 
\renewcommand{\theenumi}{\roman{enumi}}
\renewcommand{\labelenumi}{(\theenumi)}
\begin{enumerate}
\item\label{it:2} The dual subspace $V^{\bot}$ is supplementary to $V$ (so take $W$=$V^{\bot}$) and $\pi(W)=V$. 
\item\label{it:3} $V$ is self-dual and there exists a subspace $W$ supplementary to $V$ such that $\pi(W)=W$ (i.e., $\pi(W)\cap 
    V=\{0\}$). 
\end{enumerate}
\end{theorem}

\begin{proof}
Every element $x\in\fp^n$ can be uniquely written as $x=v+w$ with $v\in V$ and $w\in W$, where $W$ is a subspace supplementary to $V$. By 
Proposition~\ref{pr:MM_crit1}, $f$ belongs to the completed class $\cM$ if and only if there exists an $m$-dimensional subspace $V$ of 
$\fp^n$ such that $f$ is affine on every coset of $V$, i.e., $f(w+v)=v\cdot\pi(w)+g(w)$ for some $\pi:W\mapsto\fp^n$, $g:W\mapsto\fp$, 
where $(v,w)\in V\times W$. In the univariate form, we can also write that $f(x)=\phi(x)\cdot\pi(x)+g(x)$, where functions $\pi$ and $g$ 
(with the domain extended to $\fp^n$) are invariant by translation of their input by an element of $V$ and $\phi$ is linear with kernel 
$V^{\bot}$.

The Walsh transform of $f$ in point $b\in\fp^n$ is equal to 
\begin{align*}
W_f(b)&=\sum_{v\in V,\,w\in W}\omega^{v\cdot\pi(w)+g(w)-b\cdot(v+w)}\\
&=\sum_{w\in W}\omega^{g(w)-b\cdot w}\sum_{v\in V}\omega^{v\cdot(\pi(w)-b)}\\
&=p^m \sum_{\substack{w\in W \\ \pi(w)-b\in V^{\bot}}}\omega^{g(w)-b\cdot w}
\end{align*}
Denote $N_b=\#\{w\in W\ |\ \pi(w)-b\in V^{\bot}\}$. 

If $f$ is bent then $|W_f(b)|=p^m$ and, thus, $N_b\geq 1$ for any $b\in\fp^n$. Therefore, every $b\in\fp^n$ can be written (may be not 
uniquely) as $b=\pi(w)+v'$ for some $w\in W$ and $v'\in V^{\bot}$. Since $\#\pi(W)\leq p^m$ and $\#V^{\bot}=p^m$, we conclude that 
$\#\pi(W)=p^m$ so $\pi$ is an injection, $\pi(W)$ is a subspace supplementary to $V^{\bot}$, and thus, $N_b=1$ for any $b\in\fp^n$. 
Obviously, if $\pi(W)$ is a subspace supplementary to $V^{\bot}$ then $f$ is bent.

For the particular cases, let $b=b_1+b_2$ with $(b_1,b_2)\in V\times W$ (we select concrete $W$ later). If (\ref{it:2}) holds then taking 
$W$=$V^{\bot}$ we have $\pi(W)=V$ and 
\begin{align*}
\{w\in W\ |\ \pi(w)-b\in V^{\bot}\}&=\{w\in W\ |\ \pi(w)-b\in W\}\\
&=\{w\in W\ |\ \pi(w)-b_1\in W\}\\
&=\pi^{-1}(b_1)\enspace.
\end{align*}
Then 
\[W_f(b)=p^m\omega^{g(\pi^{-1}(b_1))-b\cdot\pi^{-1}(b_1)}=p^m\omega^{g(\pi^{-1}(b_1))-b_2\cdot\pi^{-1}(b_1)}\]
and $f$ is a regular bent function.

If (\ref{it:3}) holds then $V^{\bot}=V$ and taking $W$ such that $\pi(W)=W$, we obtain 
\begin{align*}
\{w\in W\ |\ \pi(w)-b\in V^{\bot}\}&=\{w\in W\ |\ \pi(w)-b\in V\}\\
&=\{w\in W\ |\ \pi(w)-b_2\in V\}\\
&=\pi^{-1}(b_2)\enspace.
\end{align*}
Then 
\[W_f(b)=p^m\omega^{g(\pi^{-1}(b_2))-b\cdot\pi^{-1}(b_2)}\]
and $f$ is a regular bent function.\qed
\end{proof}

\begin{note}
 \label{no:note2}
Note that for a function $f$ from the completed class of original MMF functions with associated ($m$-dimensional) subspace $V$ such that 
its $V\cap V^{\bot}=\{0\}$, condition that $\pi(V^{\bot})=V$ in (\ref{it:2}) of Theorem~\ref{th:MM_crit2} holds if and only if the 
first-order derivatives $D_u f(x)$ are balanced on $\fp^n$ for any $u\in V^*$. 

To prove this, note first that $D_u f(x)=D_u f(w+v)=u\cdot\pi(w)$, where $(v,w)\in V\times V^{\bot}$ and $x=w+v$. Thus, if 
$\pi(V^{\bot})=V$, then $D_u f(x)$ is obviously balanced. In the opposite direction, assume $u\cdot\pi(w)$ is balanced on $V^{\bot}$. 
Define a set of characters of $V$ (a finite abelian group) as $\chi_u(v)=\omega^{u\cdot v}$ for every $u\in V$. Note that $\chi_u$ is a 
trivial character if and only if $u\in V^{\bot}$ and since $V^{\bot}$ is supplementary to $V$, then $u=0$. Thus, if $u_1,u_2\in V$ with 
$u_1\ne u_2$ then $\chi_{u_1}(v)/\chi_{u_2}(v)=\chi_{u_1-u_2}(v)$ is trivial if and only if $u_1=u_2$. Therefore, all characters of $V$ are 
defined this way. Take any $v\in V$ and denote $N=\#\{w\in V^{\bot}\ |\ \pi(w)=v\}$. Then, using \cite[Eq.~(5.4)]{LiNi97}, 
\[N=\frac{1}{|V|}\sum_{w\in V^{\bot}}\sum_{\chi\in V^\wedge}\chi(\pi(w))\overline{\chi(v)}=
1+\frac{1}{|V|}\sum_{c\in V\setminus\{0\}}\overline{\chi_u(v)}\sum_{w\in V^{\bot}}\omega^{u\cdot\pi(w)}=1\enspace,\] where $V^\wedge$ is 
the group of characters of $V$ and since $|V|=|V^{\bot}|$. We conclude that $\pi(V^{\bot})=V$. 
\end{note}

Mention an interesting connection of Theorem~\ref{th:MM_crit2} to {\em the hull of a linear code} defined as $\mathrm{Hull}(V)=V\cap 
V^{\bot}$. Case (\ref{it:2}) corresponds to the Linear Complementary Dual (LCD) code with $\mathrm{Hull}(V)=\{0\}$ and case (\ref{it:3}) to 
the self-orthogonal code with $\mathrm{Hull}(V)=V$. Intermediate hulls interpolate between LCD and self-orthogonal codes where the 
dimension of $\mathrm{Hull}(V)$ falls between $0$ and $\dim(C)$. Intermediate hulls provide other cases of functions in the completed class 
$\cM$. 

\section{Binomial Bent Functions in the Completed Class $\cM$ (case $n=4k$ with $k$ odd)}
 \label{sec:MM_binom1}
In this section, we present ternary binomial vectorial bent functions in the completed class $\cM$. The approach that works here is similar 
to what was used in \cite[Proposition~2]{Do98} in the binary case (see also \cite{Le06}). We transform the function to the $4$-variable 
representation that allows to calculate the Walsh transform. 

Let $n=4k$ with $k$ odd. Take polynomial $x^4+x-1$ that is irreducible over $\fthree$, then it is also irreducible over $\fthreek$ since 
$k$ is odd (see \cite[Corollary~3.47]{LiNi97}). For $a\in\F_{3^4}$ such that $a^4+a=1$ we have $\fthreen=\fthreek(a)$. Note that $a$ is a 
primitive element of $\F_{3^4}$. In particular, any $x\in\fthreen$ can be uniquely written as 
\begin{equation}
 \label{eq:x_exp}
x=x_3 a^3+x_2 a^2+x_1 a+x_0
\end{equation}
with $x_i\in\fthreek$. Note that $\Tr^n_k(1)=1$ and using $a^4=1-a$, calculate
\begin{align*}
\Tr^n_k(a)&=a+a^{3^k}+a^{3^{2k}}+a^{3^{3k}}=a+a^3+a^9+a^{27}=0=\Tr^n_k(a^3)\\
\Tr^n_k(a^2)&=a^2+a^6+a^{18}-a^{14}=0\enspace.
\end{align*}

\begin{theorem}
 \label{th:binom1}
Let $n=4k$ with $k\equiv 3\pmod{4}$ and take $a\in\F_{3^4}$ such that $a^4+a=1$. Ternary vectorial function $F:\fthreen\mapsto\fthreek$ 
given by 
\[F(x)=\Tr^n_k\left(a x^{2(3^k+1)}+a^{-1}x^{(3^k+1)^2}\right)\]
is a quartic bent function. Moreover, component functions of $F$ are all regular bent functions from the completed class $\cM$. 
\end{theorem}

\begin{proof}
Same arguments as in the beginning of the proof of Theorem~\ref{th:binom_gf} can be used to conclude that each component functions of $F$ 
is EA-equivalent to $\Tr_k(F(x))$ but having an arbitrary sign at the second term. Also, by Note~\ref{no:note1} and since $k$ is odd, $F$ 
is EA-equivalent to $\Tr^n_k\left(a x^{2(3^k+1)}-a^{-1}x^{(3^k+1)^2}\right)$. Thus, it is sufficient to prove that ternary function 
$f(x)=\Tr_k(F(x))$ is bent. 

First, using (\ref{eq:x_exp}) convert $f(x)$ in its univariate representation into the polynomial of four variables over $\fthreek$. Note 
that for $k\equiv 3\pmod{4}$ we have that $3^k\equiv 27\pmod{80}$ and $3^{2k}\equiv 9\pmod{80}$. Then 
\begin{eqnarray}
 \label{eq:4var_rep}
\nonumber f(x)&=&\Tr_n\left(x^{2\cdot 3^k+1}\big(ax+a^{-1}x^{3^{2k}}\big)\right)\\
\nonumber&=&\Tr_n\Big((x_3 a+x_2 a^{54}+x_1 a^{27}+x_0)^2(x_3 a^3+x_2 a^2+x_1 a+x_0)\\
\nonumber&&\ \times\big(x_3 a^4+x_2 a^3+x_1 a^2+x_0 a+a^{79}(x_3 a^{27}+x_2 a^{18}+x_1 a^9+x_0)\big)\Big)\\
\nonumber&=&\Tr_k\Big(x_0^4+x_0 x_1^3+x_0^3 x_2+x_1^3 x_2-x_0 x_2^3-x_2^4-x_0^3 x_3\\
&&\ +x_0^2 x_1 x_3-x_0 x_1^2 x_3-x_0^2 x_2 x_3-x_0 x_1 x_2 x_3-x_0 x_2^2 x_3+x_0^2 x_3^2\\
\nonumber&&\ -x_0 x_1 x_3^2+x_1^2 x_3^2-x_0 x_2 x_3^2-x_1 x_2 x_3^2+x_1 x_3^3-x_3^4\Big)\enspace.
\end{eqnarray}
Also, take any $b\in\fthreen$ and expand it as $b=b_3 a^3+b_2 a^2+b_1 a+b_0$ with $b_i\in\fthreek$. Now substituting $x_0$ with $x_0+x_2$ 
we obtain 
\begin{align*}
f(x')+\Tr_n(bx')&=\Tr_k\Big(x_2\big((x_0+x_1)x_3^2-(x_0+x_1)^2 x_3-(x_0+x_1)^3+b_0+b_2\big)\\
&\ +x_0(x_0+x_1)^3-x_3(x_0-x_3)(x_0+x_1)^2+x_1 x_3^3-x_3^4\\
&\ +x_1 b_3+x_3 b_1+x_0 b_0\Big)\\
&=\Tr_k(x_2 g(x_0,x_1,x_3)+h(x_0,x_1,x_3))\enspace,
\end{align*}
where $x'=x_3 a^3+x_2 a^2+x_1 a+x_0+x_2$ and functions $g()$ and $h()$ are defined explicitly. 

Computing the Walsh transform in point $b$ we obtain
\begin{eqnarray}
 \label{eq:W_tr}
\nonumber W_f(-b)&=&\sum_{x\in\fthreen}\omega^{f(x)+\Tr_n(bx)}\\
\nonumber&=&\sum_{x_0,x_1,x_2,x_3\in\fthreek}\omega^{\Tr_k(x_2 g(x_0,x_1,x_3)+h(x_0,x_1,x_3))}\\
&=&3^k\sum_{{x_0,x_1,x_3\in\fthreek\atop g(x_0,x_1,x_3)=0}\atop x_0+x_1=0}\omega^{\Tr_k(h(x_0,x_1,x_3))}
+3^k\sum_{{x_0,x_1,x_3\in\fthreek\atop g(x_0,x_1,x_3)=0}\atop x_0+x_1\neq 0}\omega^{\Tr_k(h(x_0,x_1,x_3))}
\end{eqnarray}

Firstly, consider the case when $b_0+b_2=0$. Under this condition, $g(x_0,x_1,x_3)=0$ if and only if $x_0+x_1=0$ and, thus, the second sum 
in (\ref{eq:W_tr}) is zero. Indeed, assuming that $x_0+x_1\neq 0$ we obtain 
\[g(x_0,x_1,x_3)=(x_0+x_1)^3\left(\left(\frac{x_3}{x_0+x_1}\right)^2-\frac{x_3}{x_0+x_1}-1\right)\]
and $g(x_0,x_1,x_3)=0$ has solutions if and only if the discriminant equal to $-1$ is a square in $\fthreek$ which is false since $k$ is 
odd. Therefore, 
\begin{align*}
W_f(-b)&=3^k\sum_{x_0,x_1,x_3\in\fthreek\atop x_0+x_1=0}\omega^{\Tr_k(h(x_0,x_1,x_3))}\\
&=3^k\sum_{x_1,x_3\in\fthreek}\omega^{\Tr_k(h(-x_1,x_1,x_3))}\\
&=3^k\sum_{x_1,x_3\in\fthreek}\omega^{\Tr_k(x_1 x_3^3-x_3^4+x_1 b_3+x_3 b_1-x_1 b_0)}\\
&=3^k\sum_{x_3\in\fthreek}\omega^{\Tr_k(-x_3^4+x_3 b_1)}\sum_{x_1\in\fthreek}\omega^{\Tr_k(x_1(x_3^3+b_3-b_0))}\\
&=3^{2k}\sum_{x_3\in\fthreek\atop x_3^3=b_0-b_3}\omega^{\Tr_k(-x_3^4+x_3 b_1)}\\
&=3^{2k}\omega^{\Tr_k\left((b_0-b_3)^{3^{k-1}}(-b_0+b_1+b_3)\right)}\enspace.
\end{align*}

Secondly, consider the case when $b_0+b_2\neq 0$. Under this condition, if $x_0+x_1=0$ then $g(x_0,x_1,x_3)=b_0+b_2\neq 0$ and, thus, the 
first sum in (\ref{eq:W_tr}) is zero. We need to study solutions of $g(x_0,x_1,x_3)=0$ when $x_0+x_1\neq 0$. In this case, 
\begin{equation}
 \label{eq:g}
g(x_0,x_1,x_3)=(x_0+x_1)^3\left(\left(\frac{x_3}{x_0+x_1}\right)^2-\frac{x_3}{x_0+x_1}-1+\frac{b_0+b_2}{(x_0+x_1)^3}\right)
\end{equation}
and $g(x_0,x_1,x_3)=0$ has solutions if and only if the discriminant (denoted by $t$) and equal to $-1-\frac{b_0+b_2}{(x_0+x_1)^3}$ is a 
square in $\fthreek$. The latter identity gives that $(x_0+x_1)^3=-\frac{b_0+b_2}{t+1}$ (note that $b_0+b_2\neq 0$ as assumed and $t\neq 
-1$ since $k$ is odd). Now change the parameters that we calculate the sum over. We chose $x_0\in\fthreek$ freely and allow $t$ to run over 
squares in $\fthreek$ (including zero). Also, from (\ref{eq:g}) we obtain that if $g(x_0,x_1,x_3)=0$ then 
\[x_3=(x_0+x_1)(-1\pm\sqrt{t})=-\left(\frac{b_0+b_2}{t+1}\right)^{3^{k-1}}(-1\pm\sqrt{t})\enspace.\]
Additionally, replace $x_1$ with $x_1-x_0$. Now $x_1^3=-(b_0+b_2)/(t+1)$ and every $t$ uniquely defines $x_1$ (also 
$x_3=x_1(-1\pm\sqrt{t})$). With these new parameters, 
\begin{align*}
&h(x_0,x_1-x_0,x_3)\\
&=x_0 x_1^3-x_3(x_0-x_3)x_1^2+(x_1-x_0)x_3^3-x_3^4+(x_1-x_0)b_3+x_3 b_1+x_0 b_0\\
&=x_0\big(x_1^3-x_3 x_1^2-x_3^3-b_3+b_0\big)+x_3^2 x_1^2+x_1 x_3^3-x_3^4+x_1 b_3+x_3 b_1\\
&=x_0\big(x_1^3-(-1\pm\sqrt{t})x_1^3-x_1^3(-1\pm\sqrt{t})^3-b_3+b_0\big)+x_3^2 x_1^2+x_1 x_3^3-x_3^4+x_1 b_3+x_3 b_1\\
&=x_0\big(-x_1^3(1+t)(\pm\sqrt{t})-b_3+b_0\big)+x_3^2 x_1^2+x_1 x_3^3-x_3^4+x_1 b_3+x_3 b_1\\
&=x_0\big(\pm(b_0+b_2)\sqrt{t}-b_3+b_0\big)+x_3^2 x_1^2+x_1 x_3^3-x_3^4+x_1 b_3+x_3 b_1\enspace.
\end{align*}

Denote $\fthreek(\Box)$ the set of squares in $\fthreek$ (including zero). In the following sum, it is assumed that plus sign is taken in 
place of $\pm$ when $x_3=x_1(-1+\sqrt{t})$ and minus if $x_3=x_1(-1-\sqrt{t})$. Now we are ready to calculate 
\begin{align*}
&W_f(-b)=3^k\sum_{{x_0,x_1,x_3\in\fthreek\atop g(x_0,x_1,x_3)=0}\atop x_0+x_1\neq 0}\omega^{\Tr_k(h(x_0,x_1,x_3))}\\
&=3^k\sum_{t\in\fthreek(\Box)\atop x_3=x_1(-1\pm\sqrt{t})}\omega^{\Tr_k(x_3^2 x_1^2+x_1 x_3^3-x_3^4+x_1 b_3+x_3 b_1)}
\sum_{x_0\in\fthreek}\omega^{\Tr_k(x_0(\pm(b_0+b_2)\sqrt{t}-b_3+b_0))}\\
&=3^{2k}\omega^{\Tr_k(x_3^2 x_1^2+x_1 x_3^3-x_3^4+x_1 b_3+x_3 b_1)}\enspace,
\end{align*}
where $t=(b_3-b_0)^2/(b_0+b_2)^2$, $x_1=-\left(\frac{b_0+b_2}{t+1}\right)^{3^{k-1}}$ and $x_3=x_1(-1+(b_3-b_0)/(b_0+b_2))$. 

Since $f$ is EA-equivalent to the function
\[\Tr_k\big(x_0(x_1^3-x_1^2 x_3-x_3^3)+x_2(x_1 x_3^2-x_1^2 x_3-x_1^3) +x_1^2 x_3^2+x_1 x_3^3-x_3^4\big)\enspace,\]
it belongs to the completed class $\cM$.\qed
\end{proof}

Note that in the proof of Theorem~\ref{th:binom1} it is only required that $k$ is odd ($k\equiv 1\pmod{4}$ is also accepted). This means 
that $F$ written as a $4$-variable polynomial over $\fthreek$ in (\ref{eq:4var_rep}) is bent for any odd $k$. Remarkable that this 
polynomial over $\fthreek$ does not depend on $k$ meaning that vectorial bent function $F$ is represented by the exceptional polynomial 
over $\fthreek$. 

\section{Binomial Vectorial Bent Functions in the Completed Class $\cM$ (general case $n=4k$)}
 \label{sec:MM_binom2}
In this section, we present ternary binomial vectorial bent functions in the completed class $\cM$. The proof is based on analyzing the
first-order derivatives of a component function in the direction of elements taken from the particularly chosen $n/2$-dimensional subspace 
of $\fthreen$ and applying Theorem~\ref{th:MM_crit2}~(\ref{it:2}). We begin with a technical lemma that defines this subspace and reveals 
some relevant properties of its elements. 

\begin{lemma}
 \label{le:subsp1}
Let $n=4k$. Take any $a_1$ being a nonsquare in $\fthreen$ and define
\begin{equation}
 \label{eq:a2}
a_2=\pm I^k a_1^{(3^k+1)/2}\left((-1)^k a_1^{(3^k-1)(3^{2k}+1)/4}+a_1^{-(3^k-1)(3^{2k}+1)/4}\right)\enspace,
\end{equation}
where $I$ is a primitive $4$th root of unity in $\fthreen$ and the sign is arbitrary. There exists a unique nonzero multiplicative coset of 
$\F_{3^{2k}}$ that we denote $V$, such that for any $v\in V$, 
\begin{align}
 \label{eq:1}
\Tr^n_{2k}\big(a_2 v^{2\cdot 3^k}\big)&=0\\
 \label{eq:3}
a_1^{3^k}v^{2\cdot 3^{2k}}+a_1 v^2+a_2 v^{3^{2k}+1}&=0\enspace.
\end{align}
Moreover, such $v$ also satisfy
\begin{align}
 \label{eq:2}
\Tr^n_{2k}\big(a_1 v^{2(3^k+1)}\big)&=0\\
 \label{eq:4}
-a_1^{3^k}v^{3^k(3^k+1)}+a_2 v^{3^k+1}+a_2^{3^k}v^{3^{2k}(3^k+1)}&=0\enspace.
\end{align}
Also, the dual of $V$ is a subspace supplementary to $V$.
\end{lemma}

\begin{proof}
First, rewrite (\ref{eq:a2}) as
\begin{equation}
 \label{eq:m_a2}
a_2=\pm I^k a_1^{(3^k+1)/2} a_1^{-(3^k-1)(3^{2k}+1)/4}\big((-1)^k a_1^{(3^k-1)(3^{2k}+1)/2}+1\big)
\end{equation}
and note that the multiplicative term in brackets is a nonzero element of $\F_{3^{2k}}$ since $a_1$ is a nonsquare and 
$\big(a_1^{(3^{2k}+1)(3^k-1)/2}\big)^{3^k+1}=-1$. Therefore, 
\begin{align}
 \label{eq:mm_a2}
\nonumber a_2^{3^{2k}-1}&=a_1^{(3^k+1)(3^{2k}-1)/2} a_1^{-(3^n-1)(3^k-1)/4}\\
\nonumber &=a_1^{(3^k+1)(3^{2k}-1)/2} (-1)^{-(3^k-1)/2}\\
&=(-1)^k a_1^{(3^k+1)(3^{2k}-1)/2}\enspace.
\end{align}
It is easy to see that $a_2$ is also a nonsquare in $\fthreen$. Note that $v=0$ satisfies all four identities (\ref{eq:1})-(\ref{eq:4}) so 
we may further assume $v\neq 0$. 

Write (\ref{eq:1}) equivalently as
\begin{equation}
 \label{eq:m_2}
v^{2\cdot 3^k(3^{2k}-1)}=-a_2^{-(3^{2k}-1)}=\big(\alpha^{(3^{2k}+1)/2}a_2^{-1}\big)^{3^{2k}-1}\enspace,
\end{equation}
where $\alpha$ is a primitive element of $\fthreen$. Since $\gcd(2\cdot 3^k(3^{2k}-1),3^n-1)=2(3^{2k}-1)$ and $a_2$ (as well as $a_2^{-1}$) 
is a nonsquare, the latter equation has $2(3^{2k}-1)$ solutions in $\fthreen$ that fall into two multiplicative cosets of $\F_{3^{2k}}^*$. 
These cosets are defined by the following two identities 
\begin{equation}
 \label{eq:mm_2}
v^{3^{2k}-1}=\pm I a_2^{-3^{3k}(3^{2k}-1)/2}=\pm I a_2^{3^k(3^{2k}-1)/2}\enspace,
\end{equation}
where $I$ is a primitive $4$th root of unity in $\fthreen$.

Further, if (\ref{eq:1}) holds then
\begin{align*}
a_1^{3^{2k}-1} v^{2(3^k+1)(3^{2k}-1)}&\stackrel{(\ref{eq:m_2})}{=}a_1^{3^{2k}-1} a_2^{-3^{3k}(3^k+1)(3^{2k}-1)}\\
&=a_1^{3^{2k}-1} a_2^{(3^k-1)(3^{2k}-1)}\\
&\stackrel{(\ref{eq:mm_a2})}{=}a_1^{3^{2k}-1} a_1^{(3^{2k}-1)^2/2}\\
&=a_1^{(3^n-1)/2}=-1
\end{align*}
which means that (\ref{eq:2}) holds for every $v$ satisfying (\ref{eq:1}).

Also, if (\ref{eq:1}) holds then
\begin{align*}
&a_1+a_1^{3^k}v^{2(3^{2k}-1)}+a_2 v^{3^{2k}-1}\\
&\stackrel{(\ref{eq:mm_2})}{=}a_1-a_1^{3^k} a_2^{3^k (3^{2k}-1)}\pm I a_2 a_2^{3^k(3^{2k}-1)/2}\\
&\stackrel{(\ref{eq:m_a2},\ref{eq:mm_a2})}{=}a_1-(-1)^k a_1^{3^k} a_1^{3^k(3^k+1)(3^{2k}-1)/2}\\
&\ \pm I I^k a_1^{(3^k+1)/2} a_1^{-(3^k-1)(3^{2k}+1)/4}\left((-1)^k a_1^{(3^k-1)(3^{2k}+1)/2}+1\right)\\
&\ \times I^{k 3^k(3^{2k}-1)/2} a_1^{3^k(3^k+1)(3^{2k}-1)/4} a_1^{-3^k(3^k-1)(3^n-1)/8}\\
&=a_1+(-1)^k a_1^{3^k} a_1^{(3^k-1)(3^{2k}-1)/2}\\
&\ \pm I^{k+1} a_1^{(3^k+1)/2} a_1^{-(3^{2k}+1)(3^k-1)/4}\left((-1)^k a_1^{(3^k-1)(3^{2k}+1)/2}+1\right)\\
&\ \times a_1^{(3^k-1)(3^{2k}-1)/4}\big(a_1^{(3^n-1)/4}\big)^{(2-3^k(3^k-1))/2}\\
&=a_1+(-1)^k a_1 a_1^{(3^k-1)(3^{2k}+1)/2}\\
&\ \pm I^{k+1} a_1\left((-1)^k a_1^{(3^k-1)(3^{2k}+1)/2}+1\right) I^{(2-3^k(3^k-1))/2}\\
&=a_1+(-1)^k a_1 a_1^{(3^k-1)(3^{2k}+1)/2}\pm a_1\left((-1)^k a_1^{(3^k-1)(3^{2k}+1)/2}+1\right)
\end{align*}
is equal to zero if the sign in the last line of the equation is minus. This holds for $v$ taken from the appropriately chosen 
multiplicative coset of $\F_{3^{2k}}^*$ defined by one of the equations (\ref{eq:mm_2}). Thus, both (\ref{eq:1}) and (\ref{eq:3}) hold just 
for one of the cosets defined by (\ref{eq:mm_2}). 

Finally, if (\ref{eq:1}) and (\ref{eq:3}) hold then
\begin{align*}
&-a_1^{3^k}v^{3^{2k}-1}+a_2+a_2^{3^k}v^{(3^k+1)(3^{2k}-1)}\\
&\stackrel{(\ref{eq:3})}{=}a_2+\big(a_1 a_2^{-1}+a_1^{3^k} a_2^{-1} v^{2(3^{2k}-1)}\big)\big(a_1^{3^k}-a_2^{3^k}v^{3^k(3^{2k}-1)}\big)\\
&\stackrel{(\ref{eq:3})}{=}a_2+\big(a_1 a_2^{-1}+a_1^{3^k} a_2^{-1} v^{2(3^{2k}-1)}\big)\big(-a_1^{3^k}+a_1^{3^{2k}} v^{2\cdot 3^k(3^{2k}-1)}\big)\\
&\stackrel{(\ref{eq:mm_2})}{=}a_2-a_2^{-1}\big(a_1-a_1^{3^k} a_2^{3^k(3^{2k}-1)}\big)\big(a_1^{3^k}+a_1^{3^{2k}} a_2^{3^{2k}(3^{2k}-1)}\big)\\
&\stackrel{(\ref{eq:mm_a2})}{=}a_2-a_2^{-1} a_1^{3^k+1}\big(1+(-1)^k a_1^{(3^k-1)(3^{2k}+1)/2}\big)\big(1+(-1)^k a_1^{-(3^k-1)(3^{2k}+1)/2}\big)\\
&\stackrel{(\ref{eq:m_a2})}{=}a_2\mp I^{-k} a_1^{(3^k+1)/2} a_1^{(3^k-1)(3^{2k}+1)/4}\big(1+(-1)^k a_1^{-(3^k-1)(3^{2k}+1)/2}\big)\\
&\stackrel{(\ref{eq:a2})}{=}a_2-I^{-2k}(-1)^k a_2\\
&=0
\end{align*}
which means that (\ref{eq:4}) holds for every $v$ satisfying (\ref{eq:1}) and (\ref{eq:3}). 

Denoting the dual of $V$ as $V^{\bot}$ , it is supplementary to $V$ if and only if $V\cap V^{\bot}=\{0\}$. Assume there exists a nonzero 
$c\in V$ such that $\Tr_n(cv)=0$ for any $v\in V$ (i.e., $c\in V^{\bot}$). Since $c,v\in V$ being the multiplicative coset of 
$\F_{3^{2k}}$, we have $c^{-1} v\in\F_{3^{2k}}$ and 
\begin{align*}
\Tr_n(cv)&=\Tr_{2k}\big(c^{-1} v\Tr_{2k}^n(c^2)\big)\\
&=\Tr_{2k}\big(c v(1+c^{2(3^{2k}-1)})\big)\\
&\stackrel{(\ref{eq:mm_2})}{=}\Tr_{2k}\big(c v(1-a_2^{3^k(3^{2k}-1)})\big)
\end{align*}
is zero for any $v\in V$ if and only if $a_2^{3^{2k}-1}=1$ which is impossible since $a_2$ is a nonsquare in $\fthreen$.\qed 
\end{proof}

\begin{theorem}
 \label{th:binom_gf}
Let $n=4k$. Take any $a_1$ being a nonsquare in $\fthreen$ and define $a_2$ according to (\ref{eq:a2}). Ternary vectorial function 
$F:\fthreen\mapsto\fthreek$ given by 
\begin{equation}
 \label{eq:binom}
F(x)=\Tr^n_k\left(a_1 x^{2(3^k+1)}+a_2 x^{(3^k+1)^2}\right)
\end{equation}
is a quartic bent function. Moreover, component functions of $F$ are all regular bent functions from the completed class $\cM$. 
\end{theorem}

\begin{proof}
We have to show that component functions 
\[f_{\alpha}(x)=\Tr_n\left(\alpha a_1 x^{2(3^k+1)}+\alpha a_2 x^{(3^k+1)^2}\right)\]
are bent for any $\alpha\in\fthreek^*$. Let $\xi$ be a primitive element of $\fthreen$ and $\alpha=\xi^{i(3^{2k}+1)(3^k+1)}$ for some 
$i=0,1,2,\dots$. Then 
\begin{align*}
f_{\alpha}(x)&=\Tr_n\left(a_1(\xi^{i(3^{2k}+1)/2}x)^{2(3^k+1)}+\alpha a_2 x^{(3^k+1)^2}\right)\\
&=\Tr_n\left(a_1 y^{2(3^k+1)}+(-1)^i a_2 y^{(3^k+1)^2}\right)\enspace,
\end{align*}
where $y=\xi^{i(3^{2k}+1)/2}x$. Thus, each component functions of $F$ is EA-equivalent to $\Tr_k(F(x))$, sometimes with $a_2$ changing the 
sign to opposite. It is sufficient to prove that ternary function $f(x)=\Tr_k(F(x))$ is bent when the sign of $a_2$ is arbitrary as in 
(\ref{eq:a2}). 

Now take any $c\in\fthreen$ and $v\in V$, where $V$ is defined in Lemma~\ref{le:subsp1}. Using a computer algebra package, we obtain that 
\begin{align*}
&f(c+v)=\\
&=f(c)+\Tr_n\big(c^{2\cdot 3^k}(a_2 c^{3^{2k}}-a_1 c)v-c^{3^k+1}(a_2 c^{3^{2k}}+a_1 c)v^{3^k}+a_2 c^{2\cdot 3^k+1}v^{3^{2k}}\\
&\ +(a_1 c^{3^k+1}-a_2 c^{{3^k(3^k+1)}})v^{3^k+1}-a_2 c^{3^k+1}v^{3^k(3^k+1)}\\
&\ +a_1 c^{2\cdot 3^k}v^2+(a_1 c^2+a_2 c^{3^{2k}+1})v^{2\cdot 3^k}+(a_2 c^{3^{2k}}-a_1 c)v^{2\cdot 3^k+1}\\ 
&\ +a_2 c^{2\cdot 3^k}v^{3^{2k}+1}-a_2 c^{3^k}v^{3^k(3^k+1)+1}-a_1 c^{3^k}v^{3^k+2}+a_2 cv^{3^k(3^k+2)}\\
&\ +a_2 v^{(3^k+1)^2}+a_1 v^{2(3^k+1)}\big)\\
&=f(c)+\Tr_n\big(c^{2\cdot 3^k}(a_2 c^{3^{2k}}-a_1 c)v-c^{3^k+1}(a_2 c^{3^{2k}}+a_1 c)v^{3^k}+a_2 c^{2\cdot 3^k+1}v^{3^{2k}}\\
&\ +(a_1^{3^k}v^{3^k(3^k+1)}-a_2^{3^k}v^{3^{2k}(3^k+1)}-a_2 v^{3^k+1})(c^{3^k(3^k+1)}-c^{3^{2k}}v^{3^k}-c^{3^k} v^{3^{2k}})\\
&\ +(a_1^{3^k}v^{2\cdot 3^{2k}}+a_1 v^2+a_2 v^{3^{2k}+1})c^{2\cdot 3^k}\big)\\ 
&\ +\Tr_{2k}\big((c^{3^{2k}+1}+v^{3^{2k}+1})\Tr^n_{2k}(a_2 v^{2\cdot 3^k})+\Tr^n_{2k}(a_1 v^{2(3^k+1)})\big)\\
&=f(c)+\Tr_n\big(c^{2\cdot 3^k}(a_2 c^{3^{2k}}-a_1 c)v-c^{3^k+1}(a_2 c^{3^{2k}}+a_1 c)v^{3^k}+a_2 c^{2\cdot 3^k+1}v^{3^{2k}}\big)\enspace,
\end{align*}
where in the last identity, we used (\ref{eq:1}), (\ref{eq:3}), (\ref{eq:2}), and (\ref{eq:4}). We conclude that $f$ is affine on the 
cosets of $V$ and, thus, belongs to the completed class of MMF functions with associated $2k$-dimensional subspace $V$. Therefore, the 
second-order derivatives of $f$ in the direction of elements of $V$ vanish (by Proposition~\ref{pr:MM_crit1}).

Now prove that the first-order derivatives of $f$ in the direction of elements $v\in V^*$, where $V$ is defined in Lemma~\ref{le:subsp1}, 
are balanced (see Note~\ref{no:note2}). Using the earlier computed expression for $f(c+v)$, we obtain
\begin{align*}
&D_v f(c)=f(c+v)-f(c)\\
&=\Tr_n\big(a_2 v c^{3^k(3^k+2)}-a_2 v^{3^k} c^{3^k(3^k+1)+1}-(a_1 v-a_2 v^{3^{2k}})c^{2\cdot 3^k+1}-a_1 v^{3^k} c^{3^k+2}\big)\enspace.
\end{align*}

Note that for any $u\in V$, 
\[D_v f(c+u)=f(c+v+u)-f(c+u)=D_{v+u}f(c)-D_u f(c)=D_v f(c)\]
and, thus, $D_v f(c)$ is constant on the cosets of $V$. Therefore, to prove that $D_v f(c)$ is balanced, it is sufficient to take $c$ 
running over a supplementary subspace of $V$ that we denote $W$. Assume $\alpha$ is a primitive element of $\fthreen$. Here two cases are 
possible: if $V=\F_{3^{2k}}$ then we take $W=\alpha^{(3^{2k}+1)/2}\F_{3^{2k}}$; otherwise, $V\cap\F_{3^{2k}}=\{0\}$ and we take 
$W=\F_{3^{2k}}$. 

First, if $v\in\F_{3^{2k}}^*$ and $c\in\alpha^{(3^{2k}+1)/2}\F_{3^{2k}}$ then $c^{3^{2k}}=-c$ and 
\begin{align*}
D_v f(c)&=\Tr_n\Big(v\big(a_1^{3^k} - a_1 - a_2^{3^k}\big)c^{2\cdot 3^k+1}\Big)\\
&\stackrel{(\ref{eq:3},\ref{eq:2})}{=}\Tr_n\big(va_2^{3^k}c^{2\cdot 3^k+1}\big)\\
&\stackrel{(\ref{eq:1})}{=}-\Tr_{2k}\big(va_2^{3^k}c^{2\cdot 3^k+1}\big)\enspace,
\end{align*}
since in this case, (\ref{eq:2}) gives $a_1^{3^{2k}}=-a_1$ and (\ref{eq:1}) gives $a_2^{3^{2k}}=-a_2$. Note that since 
\begin{align*}
\gcd(2\cdot 3^k+1,3^{2k}-1)&=\gcd(2\cdot 3^k+1,2\cdot 3^{2k}-2)\\
&=\gcd(2\cdot 3^k+1,-2-3^k)\\
&=\gcd(-3,-2-3^k)=1\enspace,
\end{align*}
function $c^{2\cdot 3^k+1}$ is a one-to-one map of $\alpha^{(3^{2k}+1)/2}\F_{3^{2k}}$ onto $\alpha^{(3^{2k}+1)(2\cdot 3^k+1)/2}\F_{3^{2k}}$ 
and $va_2^{3^k}c^{2\cdot 3^k+1}$ runs over all the elements in $\F_{3^{2k}}$. 

In the other case when $c\in\F_{3^{2k}}$, we have
\begin{align*}
&D_v f(c)=\Tr_n\big((a_2 v - a_2^{3^k} v^{3^{2k}} - a_1 v + a_2 v^{3^{2k}} - a_1^{3^k} v^{3^{2k}})c^{2\cdot 3^k+1}\big)\\
&\stackrel{(\ref{eq:3},\ref{eq:4})}{=}\Tr_n\Big(v\big(a_1^{3^k}v^{2(3^{2k}-1)} - a_2 v^{3^{2k}-1} - a_2^{3^k}v^{(3^k+1)(3^{2k}-1)} - a_2^{3^k}v^{3^{2k}-1}\big)c^{2\cdot 3^k+1}\Big)\\
&\stackrel{(\ref{eq:4})}{=}\Tr_n\Big(v\big(a_2^{3^k}v^{(3^k+2)(3^{2k}-1)} - a_2^{3^k}v^{(3^k+1)(3^{2k}-1)} - a_2^{3^k}v^{3^{2k}-1}\big)c^{2\cdot 3^k+1}\Big)\\
&=\Tr_n\Big(a_2^{3^k}v^{3^{2k}}\big(v^{(3^k+1)(3^{2k}-1)} - v^{3^k(3^{2k}-1)} - 1\big)c^{2\cdot 3^k+1}\Big)\\
&=\Tr_{2k}\Big(\Tr^n_{2k}\big(a_2^{3^k}v^{3^{2k}}\big(v^{(3^k+1)(3^{2k}-1)} - v^{3^k(3^{2k}-1)} - 1\big)\big)c^{2\cdot 3^k+1}\Big)\enspace.
\end{align*}
Note that since $\gcd(2\cdot 3^k+1,3^{2k}-1)=1$, function $c^{2\cdot 3^k+1}$ is a permutation over $\F_{3^{2k}}$. It remains to show that 
\[\Tr^n_{2k}\big(a_2^{3^k}v^{3^{2k}}\big(v^{(3^k+1)(3^{2k}-1)} - v^{3^k(3^{2k}-1)} - 1\big)\big)\neq 0\enspace.\]
If the opposite holds then
\[a_2^{3^k}v^{3^{2k}}\big(v^{(3^k+1)(3^{2k}-1)} - v^{3^k(3^{2k}-1)} - 1\big)+a_2^{3^{3k}}v\big(v^{(3^k+1)(3^{2k}-1)} - v^{3^k(3^{2k}-1)} - 1\big)^{3^{2k}}=0\]
Denoting $v^{3^{2k}-1}=C\ne 0$ and using $a_2^{3^{3k}}=-a_2^{3^k}C^2$ (which follows from (\ref{eq:1})) we obtain 
\begin{align*}
&a_2^{3^k}C\big(C^{3^k+1} - C^{3^k} - 1\big)-a_2^{3^k}C^2\big(C^{3^k+1} - C^{3^k} - 1\big)^{3^{2k}}\\
&=a_2^{3^k}C\big(C^{3^k+1} - C^{3^k} - 1\big)-a_2^{3^k}C\big(C^{-3^k} - C^{-3^k+1} - C\big)=0
\end{align*}
that is equivalent to
\[C^{3^k+1} - C^{3^k} - 1 - C^{-3^k} + C^{-3^k+1} + C=C^{-3^k}\big(C-1\big)^{2\cdot 3^k+1}=0\]
which is impossible since $C\ne 1$ as $v\notin\F_{3^{2k}}$.

We also obtain that
\[D_v f(c)=\Tr_{2k}\Big(a_2^{3^k}v^{(3^{2k}-1)(1-3^k)+1}\big(v^{3^{2k}-1} - 1\big)^{2\cdot 3^k+1}c^{2\cdot 3^k+1}\Big)\enspace.\]

This proves that the first-order derivatives of $f$ in the direction of elements $v\in V^*$ are balanced and it remains to apply 
Theorem~\ref{th:MM_crit2}~(\ref{it:2}).\qed 
\end{proof}

\begin{note}
 \label{no:note1}
Take $n=4k$ with $k$ odd and function $F$ having the form of (\ref{eq:binom}) with arbitrary $a_1,a_2\in\fthreen$. Let $\xi$ be a primitive 
element of $\fthreen$. Substitute $x$ in $F(x)$ by $\xi^{(3^n-1)/16} x$. Since $3^k+1\pmod{8}=4$ then we obtain function $-G(x)$ where 
\[G(x)=\Tr^n_k\left(a_1 x^{2(3^k+1)}-a_2 x^{(3^k+1)^2}\right)\enspace.\]
Therefore, functions $F(x)$ and $G(x)$ are EA-equivalent and and either both are bent or not.
\end{note}

\begin{note}
Following are some facts on the number of EA-inequivalent ternary bent functions, components of vectorial function (\ref{eq:binom}), for 
small values of $n$. These upper bounds are obtained using considerations that eliminate obvious cases when equivalence holds. For $n=4$, 
component functions of (\ref{eq:binom}) are of Dillon type and there are at most $2$ inequivalent cases. For $n=8$ and for $n=12$, there 
are at most $6$ EA-inequivalent ternary bent functions of this type; and at most $22$ for $n=16$. 
\end{note}

\section{Trinomial Bent Functions in the Completed Class $\cM$}
 \label{sec:MM_trinom}
In this section, we present ternary trinomial bent functions in the completed class $\cM$. The proof is based on analyzing the first-order 
derivatives derivatives in the direction of elements taken from the particularly chosen $n/2$-dimensional subspace of $\fthreen$ and 
applying Theorem~\ref{th:MM_crit2}. We begin with a technical lemma that defines this subspace and reveals some relevant properties of its 
elements. 

\begin{lemma}
 \label{le:subsp2}
Let $n=2k$ with $k>1$. Also let $a_1=a_3=1$, $a_2=\pm I$ for $k$ odd and $a_1=\alpha^{3^k+2}$, $a_3=\alpha$,
\begin{equation}
 \label{eq:tri_a2}
a_2=\pm Ia_1^{-(3^k-3)/2}=\mp Ia_3^{(3^k+5)/2}
\end{equation}
for $k$ even, where $\alpha$ is a primitive element of $\fthreen$ and $I$ is a primitive $4$th root of unity in $\fthreen$ and the sign of 
$a_2$ is arbitrary. 

There exists a unique nonzero multiplicative coset of $\fthreek$ that we denote $V$, such that for any nonzero $v\in V$, 
\begin{equation}
 \label{eq:V_def}
a_1 v^{3^k-1}=a_2\enspace.
\end{equation}
Moreover, for such $v$,
\begin{equation}
 \label{eq:traces}
\Tr_k^n(a_1 v^2)=\Tr_k^n(a_2 v^4)=\Tr_k^n(a_3v^2)=0\enspace.
\end{equation}
Also, $V\cap\fthreek=\{0\}$ and the dual of $V$ is a subspace supplementary to $V$ when $k$ is even and $V$ is self-dual when $k$ is odd. 
\end{lemma}

\begin{proof}
Rewrite (\ref{eq:V_def}) as
\begin{equation}
 \label{eq:ca1}
v^{3^k-1}=a_1^{-1}a_2\stackrel{(\ref{eq:tri_a2})}{=}\pm Ia_1^{-(3^k-1)/2}
\end{equation}
and observe that solution for $v$ in the form of a nonzero multiplicative coset of $\fthreek^*$ exists if $3^k-1$ divides 
$(3^n-1)/4-i(3^k-1)/2$, where $a_1=\alpha^i$. Since $a_1$ and $a_3$ have the same parity defined by the parity of $k$, this exactly allows 
the above divisibility to hold. The coset defined in (\ref{eq:V_def}) is unique. 

Now check that identities (\ref{eq:traces}) hold. First, note that
\begin{align}
 \label{eq:tri_a12}
\nonumber a_2^{3^k-1}&\stackrel{(\ref{eq:tri_a2})}{=}I^{3^k-1}a_1^{-(3^k-3)(3^k-1)/2}\\
&=(-1)^k a_1^{-(3^n-1)/2+2(3^k-1)}=-a_1^{2(3^k-1)}\enspace.
\end{align}
Then
\begin{align*}
&Tr_k^n(a_1v^2)=a_1 v^2+a_1^{3^k} v^{2\cdot3^k}\stackrel{(\ref{eq:ca1})}{=}a_1 v^2-a_1 v^2=0\\
&\Tr_k^n(a_2v^4)=a_2v^4+a_2^{3^k}v^{4\cdot3^k}\stackrel{(\ref{eq:ca1})}{=}a_2v^4+a_2^{3^k}a_1^{-2(3^k-1)}v^4\stackrel{(\ref{eq:tri_a12})}{=}a_2v^4-a_2v^4=0\\
&\Tr_k^n(a_3v^2)=a_3v^2\big(1+a_3^{3^k-1}v^{2(3^k-1)}\big)\stackrel{(\ref{eq:ca1})}{=}a_3v^2\big(1-a_3^{3^k-1}a_1^{-(3^k-1)}\big)=0\enspace.\\
\end{align*}

Assume there exists a nonzero $v\in V\cap\fthreek$, then by (\ref{eq:V_def}), $a_1=a_2$ and (\ref{eq:tri_a2}) gives that 
$a_1^{(3^k-1)/2}=\pm I$, thus, $a_1^{3^k-1}=-1$ that contradicts definition of $a_1$. 

Denoting the dual of $V$ as $V^{\bot}$, it is supplementary to $V$ if and only if $V\cap V^{\bot}=\{0\}$. Assume there exists a nonzero 
$c\in V$ such that $\Tr_n(cv)=0$ for any $v\in V$ (i.e., $c\in V^{\bot}$). Since $c,v\in V$ being the multiplicative coset of $\F_{3^k}$, 
we have $c^{-1} v\in\F_{3^k}$ and 
\begin{align*}
\Tr_n(cv)&=\Tr_k\big(c^{-1} v\Tr_k^n(c^2)\big)\\
&=\Tr_k\big(cv(1+c^{2(3^k-1)})\big)\\
&\stackrel{(\ref{eq:V_def})}{=}\Tr_k\big(cv(1+a_1^{-2}a_2^2)\big)\\
&\stackrel{(\ref{eq:tri_a2})}{=}\Tr_k\big(cv(1-a_1^{-(3^k-1)})\big)
\end{align*}
is zero for any $v\in V$ if and only if $a_1^{3^k-1}=1$ which is impossible for $k$ even. For $k$ odd, $V$ is self-dual.\qed
\end{proof}

\begin{theorem}
 \label{th:trinom_M1}
Let $n=2k$ with $k>1$ and not divisible by four. Define coefficients $a_1$, $a_2$ and $a_3$ as in Lemma~\ref{le:subsp2}. Ternary function 
$f:\fthreen\mapsto\fthree$ given by 
\begin{equation}
 \label{eq:trinom}
f(x)=\Tr_n\left(a_1 x^{2(3^k+2)} + a_2 x^{3^k+5} + a_3 x^2\right)
\end{equation}
is a quartic regular bent function from the completed class $\cM$.
\end{theorem}

\begin{proof}
Take any $c\in\fthreen$ and $v\in V$, where $V$ is defined in Lemma~\ref{le:subsp2}. Using a computer algebra package, we obtain that 
\begin{align*}
&f(c+v)=f(c)\\
&+\Tr_n\big((a_1 c^{2\cdot 3^k+3} - a_2 c^{3^k+4} - a_3 c)v + c^{3^k+1}(a_1 c^{3^k} + a_2 c)v^3 - (a_1 c^{3^k+4} - a_2 c^5)v^{3^k}\\
&+(a_2 c^{3^k+3} + a_3)v^2 + (a_1 c^{2\cdot 3^k} - a_2 c^{3^k+1})v^4 + a_2 c^{3^k} v^5 - (a_2 c^4 + a_1 c^{3^k+3})v^{3^k+1}\\
&+ a_2 c^3 v^{3^k+2} - (a_1 c^{3^k+1} - a_2 c^2)v^{3^k+3} - (a_1 c^{3^k} + a_2 c)v^{3^k+4} + a_2 v^{3^k+5}\\
&+ a_1 c^4 v^{2\cdot 3^k} + a_1 c^3 v^{2\cdot 3^k+1} + a_1 c v^{2\cdot 3^k+3} + a_1 v^{2(3^k+2)}\big)\\
&\stackrel{(\ref{eq:V_def})}{=}f(c) + \Tr_n\big(a_1 c^{2\cdot 3^k+3}v +a_2 c^5 v^{3^k} - c^{3^k+4}(a_2 v + a_1 v^{3^k})\\
&+a_2 c^{3^k+3} v^2 + a_1 c^{2\cdot 3^k} v^4 + a_2 c^{3^k} v^5 - a_2 c^4 v^{3^k+1} - a_2 c^{3^k+3} v^2\\
&+ a_1 c^3 v^{2\cdot 3^k+1} + a_1 c^2 v^{2(3^k+1)} - a_2 c^{3^k} v^5 - a_2 c v^{3^k+4}\\
&+ a_2 c^4 v^{3^k+1} + a_1 c^3 v^{2\cdot 3^k+1} + a_2 c v^{3^k+4} - a_3 c v\big) + \Tr_k\big(c^{3^k+1}\Tr_n^k((a_1 c^{3^k} + a_2 c)v^3)\\
&+(v^{2(3^k+1)} - (c v)^{3^k+1})\Tr^n_k(a_1 v^2) + (v^{3^k+1} - c^{3^k+1})\Tr^n_k(a_2 v^4) + \Tr^n_k(a_3 v^2)\big)\\
&\stackrel{(\ref{eq:V_def},\ref{eq:traces})}{=}f(c) + \Tr_n\big(a_1 c^{2\cdot 3^k+3}v +a_2 c^5 v^{3^k} + a_2 c^{3^k+4}v\big)\\
&+ a_1 c^{2\cdot 3^k} v^4 - a_1 c^3 v^{2\cdot 3^k+1} - a_3 c v + a_1 c^2 v^{2(3^k+1)}\big)\\
&=f(c) + \Tr_n\big(a_1 c^{2\cdot 3^k+3}v + a_2 c^{3^k+4}v + a_2 c^5 v^{3^k}\big)
\end{align*}
since
\begin{align*}
&\Tr^n_k\big(a_1 v^3 c^{3^k}+a_2 v^3 c\big)\\
&\stackrel{(\ref{eq:V_def})}{=}v^3\big(a_1 c^{3^k}+a_2 c+a_1^{3^k-3}a_2^3 c+a_1^{-3}a_2^{3^k+3} c^{3^k}\big)\\
&\stackrel{(\ref{eq:tri_a2})}{=}v^3\big(a_1 c^{3^k}+a_2 c-a_2 c-a_1 c^{3^k}\big)=0\enspace.
\end{align*}
and
\begin{align*}
\Tr_n\big(a_1 c^3 v^{2\cdot 3^k+1} + a_3 c v\big)&=\Tr_n\big(c^3 v^3 (a_1 v^{2(3^k-1)} + a_3^3)\big)\\
&\stackrel{(\ref{eq:V_def})}{=}\Tr_n\big(c^3 v^3 (a_1^{-1} a_2^2 + a_3^3)\big)\\
&\stackrel{(\ref{eq:tri_a2})}{=}\Tr_n\big(c^3 v^3 (-a_1^{-(3^k-2)} + a_3^3)\big)=0\enspace.
\end{align*}
and
\begin{align*}
\Tr_n\big(a_1 c^{2\cdot 3^k} v^4 + a_1 c^2 v^{2(3^k+1)}\big)&=\Tr_n\big(c^{2} v^{2(3^k+1)}(a_1^{3^k} v^{2(3^k-1)} + a_1)\big)\\
&\stackrel{(\ref{eq:V_def})}{=}\Tr_n\big(c^{2} v^{2(3^k+1)}(a_1^{3^k-2} a_2^2 + a_1)\big)\stackrel{(\ref{eq:tri_a2})}{=}0\enspace.
\end{align*}

We conclude that $f$ is affine on the cosets of $V$ and, thus, belongs to the completed class of MMF functions with associated 
$k$-dimensional subspace $V$. Therefore, the second-order derivatives of $f$ in the direction of elements of $V$ vanish (by 
Proposition~\ref{pr:MM_crit1}). 

Now prove that the first-order derivatives of $f$ in the direction of elements $v\in V^*$, where $V$ is defined in Lemma~\ref{le:subsp2}, 
are balanced. For $k$ odd, we prove a stronger property that $D_uf(v+w)=\Tr_n(u\pi(w))$ with $v\in V$ and permutation 
$\pi:\fthreek\mapsto\fthreek$, and apply Theorem~\ref{th:MM_crit2}~(\ref{it:3}). Using the earlier computed expression for $f(c+v)$, we 
obtain 
\begin{align*}
&D_v f(c)=f(c+v)-f(c)\\
&=\Tr_n\big(a_1 c^{2\cdot 3^k+3}v + a_2 c^{3^k+4}v + a_2 c^5 v^{3^k}\big)\\
&=\Tr_k\big(\big(a_2 c^4+a_2^{3^k}c^{4\cdot 3^k}\big)\big(v c^{3^k}+v^{3^k}c\big)+c^{2(3^k+1)}\big(a_1 v c+a_1^{3^k}v^{3^k}c^{3^k}\big)\big)\enspace.
\end{align*}

Note that for any $u\in V$, 
\[D_v f(c+u)=f(c+v+u)-f(c+u)=D_{v+u}f(c)-D_u f(c)=D_v f(c)\]
and, thus, $D_v f(c)$ is constant on the cosets of $V$. Therefore, to prove that $D_v f(c)$ is balanced, it is sufficient to take $c$ 
running over a supplementary subspace of $V$. By Lemma~\ref{le:subsp2}, $V\cap\fthreek=\{0\}$ so we can assume $c\in\fthreek$ and obtain 
\[D_vf(c)=\Tr_k\big(\big(a_1 v + a^{3^k}_1 v^{3^k}+\big(v+v^{3^k}\big)\big(a_2+a_2^{3^k}\big)\big)c^5\big)\enspace.\]
Conditions on $k$ provide that $\gcd(3^k-1,5)=1$, so for proving that $D_vf(c)$ is balanced on $\fthreek$ we just need to check that 
coefficient at $c^5$ (that permutes $\fthreek$) is not zero. 

First, taking $k$ odd we have $a_1=1$, $a_2=\pm I$ so $a_2+a_2^{3^k}=0$ and $v+v^{3^k}=v(1\pm I)\neq 0$, then $D_u f(v+w)=D_u f(w)=\Tr_n(u 
w^5)$ with $w\in\fthreek$ and Theorem~\ref{th:MM_crit2}~(\ref{it:3}) can be applied. Finally, when $k$ is even, 
\begin{align*}
&a_1 v + a^{3^k}_1 v^{3^k}+\big(v+v^{3^k}\big)\big(a_2+a_2^{3^k}\big)=v a_1^{-1}\big(a_1^2 + a^{3^k}_1 a_2+(a_1+a_2)\big(a_2+a_2^{3^k}\big)\big)\\
&=v a_1^{-1}(a_1-a_2)\big(a_2^{3^k}+a_1-a_2+a_2(a_1-a_2)^{3^k-1}\big)\enspace.
\end{align*}
Denote $y=a_1-a_2$ that is not zero and solve 
\begin{equation}
 \label{eq:nz}
a_2^{3^k}+y=-a_2 y^{3^k-1}
\end{equation}
for the unknown $y$. Raising both sides of (\ref{eq:nz}) to the power of $3^k+1$ results in $a_2^{3^k}y^{3^k}+a_2 y+y^{3^k+1}=0$. Make 
substitution $y=-a_2^{3^k}(z+1)$ in the latter identity to obtain that $z^{3^k+1}=1$ and $z\neq -1$. Multiplying both sides of 
(\ref{eq:nz}) by $y\neq 0$ and using the substitution for $y$ gives 
\[-a_2^2(z^{3^k}+1)+a_2^{2\cdot 3^k}(z^2-z+1)-a_2^{2\cdot 3^k}(z+1)=-a_2^2(z^{-1}+1)+a_2^{2\cdot 3^k}(z^2+z)=0\enspace.\] 
Since $z\neq -1$, equation (\ref{eq:nz}) becomes $-a_2^2+a_2^{2\cdot 3^k}z^2=0$ giving 
\[a_2^{3^k}z=-y-a_2^{3^k}=a_2-a_1-a_2^{3^k}=\pm a_2\enspace.\]
If $a_2-a_1-a_2^{3^k}=a_2$ then $a_1+a_2^{3^k}=0$ that is impossible from how $a_1$ and $a_2$ are defined. If $a_2-a_1-a_2^{3^k}=-a_2$ then 
$a_1+a_2+a_2^{3^k}=0$ that is also impossible since $a_1\notin\fthreek$.\qed
\end{proof}

\begin{note}
Take $n=2k$ with $k$ odd and function $f$ having the form of (\ref{eq:trinom}) with arbitrary $a_1,a_2,a_3\in\fthreen$. Substitute $x$ in 
$f(x)$ by $I x$, where $I$ is a primitive $4$th root of unity in $\fthreen$. Since $3^k\pmod{4}=3$ then we obtain function $-g(x)$ where 
\[g(x)=\Tr_n\big(a_1 x^{2\cdot3^k+4} - a_2 x^{3^k+5} + a_3 x^2\big)\enspace.\]
Therefore, functions $f(x)$ and $g(x)$ are EA-equivalent and either both are bent or not. This means that for odd $k$, changing sign in 
(\ref{eq:tri_a2}) gives equivalent bent functions. 
\end{note}

\begin{theorem}
 \label{th:trinom_M2}
Let $n=2k$ with $k>1$ and not divisible by six. Ternary vectorial function $F:\fthreen\mapsto\fthreek$ given by 
\[F(x)=\Tr^n_k\big(a(x^{2(3^{k+1}+1)}+x^{3^{k+1}+5}+x^8)\big)\enspace,\]
where $a=\xi^{(3^k+1)/2}$ and $\xi$ is a primitive element of $\fthreen$, is a quartic bent function. Moreover, component functions of $F$ 
are all regular bent functions from the completed class $\cM$.
\end{theorem}

\begin{proof} 
First, note that $a^{3^k}=-a$, also, since $3^6\pmod{7}=1$ and $3^i\pmod{7}\neq 1$ for $i=1,\dots,5$, we have that $3^k-1\not\equiv 
0\pmod{7}$ if $k$ is not divisible by six. 

To prove bentness, we show that all component functions of $F$ are bent. For any $\alpha\in\fthreek^*$, denote $f_{\alpha}(x)=\Tr_k(\alpha 
F(x))$ a component function of $F$. Now take any $c\in\fthreen$ and $v\in\fthreek$ to obtain 
\begin{align*}
&f_{\alpha}(c+v)=\Tr_n\big(\alpha a(v+c)^{3^{k+1}+2}((v+c)^{3^k}+v+c)^3+\alpha a(v+c)^8\big)\\
&=\Tr_n\big(\alpha a(c^{2\cdot 3^{k+1}}v^2-c^{2\cdot 3^{k+1}+1}v+c^{2(3^{k+1}+1)}\\
&\ +c^{3(3^k+1)}v^2-c^{3^{k+1}+4}v+c^{3^{k+1}+5}+ac^6 v^2-c^7 v+c^8)\big)\\
&=\Tr_n\big(-\alpha a(c^{2\cdot 3^{k+1}+1}+c^{3^{k+1}+4}+c^7)v+\alpha a(c^{2(3^{k+1}+1)}+c^{3^{k+1}+5}+c^8)\big)
\end{align*}
proving that $f_{\alpha}$ is affine on the cosets of the $k$-dimensional subspace $V=\fthreek$. Thus, by Proposition~\ref{pr:MM_crit1}, 
$f_{\alpha}$ belongs to the completed class of MMF functions with associated subspace $V$. 

Note that $W=aV$ is a subspace supplementary to $V$ that is also dual to $V$ since $\Tr^n_k(a)=0$. Denoting $\psi(w)=-\alpha a(w^{2\cdot 
3^{k+1}+1}+w^{3^{k+1}+4}+w^7)$, it is easy to see that $\psi$ is a one-to-one mapping of $W$ onto $V$ if $\gcd(7,3^k-1)=1$. Write the 
elements $x\in\fp^n$ uniquely as $x=v+w$ with $v\in V$ and $w\in W$. For any $d\in V^*$, the first-order derivative $D_d 
f_{\alpha}(v+w)=\Tr_n(\psi(w)d)=-\Tr_k(\psi(w)d)$ is balanced on $\fthreen$. By Theorem~\ref{th:MM_crit2}~(\ref{it:2}), $f_{\alpha}$ 
belongs to the completed class $\cM$.\qed 

\end{proof} 

\section{Exceptional Polynomials Giving Vectorial Bent Functions in Class $\cM$}
 \label{sec:MM_excep}
In this section, we present particular subclasses of vectorial bent functions (\ref{eq:binom}) that are in class $\cM$ and can be 
represented by exceptional polynomials over $\fthreek$. The technique used is the same as in Section~\ref{sec:MM_binom1}. 

Let $n=4k$ with $k$ odd. Take polynomial $x^4-x^2-1$ that is irreducible (of order $16$) over $\fthree$, then it is also irreducible over 
$\fthreek$ since $k$ is odd (see \cite[Corollary~3.47]{LiNi97}). For $a\in\F_{3^4}$ such that $a^4-a^2=1$ we have $\fthreen=\fthreek(a)$. 
Note that $a$ is an element of order $16$ in $\F_{3^4}$. In particular, any $x\in\fthreen$ can be uniquely written as (\ref{eq:x_exp}). 
Note that $\Tr^n_k(1)=1$ and using $a^4=a^2+1$, calculate
\begin{align*}
\Tr^n_k(a)&=a+a^{3^k}+a^{3^{2k}}+a^{3^{3k}}=a+a^3+a^9+a^{11}=0=\Tr^n_k(a^3)\\
\Tr^n_k(a^2)&=a^2+a^6+a^{18}-a^{14}=-a^6-a^2=-1\enspace.
\end{align*}

\begin{theorem}
 \label{th:binom2}
Let $n=4k$ with $k$ odd. Further, select $a\in\F_{3^4}$ such that $a^4-a^2-1=0$. Also let $a_1=a$ and $a_2=a^5$ (resp., $a_1=a^5$ and 
$a_2=a$) if $k\equiv 1\pmod{4}$ (resp., $k\equiv 3\pmod{4}$). Ternary vectorial function $F:\fthreen\mapsto\fthreek$ in (\ref{eq:binom}) is 
a quartic bent function having components in class $\cM$. 
\end{theorem}

\begin{proof}
Take any $\alpha\in\fthreek$ and using (\ref{eq:x_exp}), convert $\Tr_k(\alpha F(x))$ in its univariate representation into the polynomial 
of four variables over $\fthreek$. Note that $3^k\equiv 3\pmod{80}$ (resp., $3^k\equiv 27\pmod{80}$) for $k\equiv 1\pmod{4}$ (resp., 
$k\equiv 3\pmod{4}$) and $3^{2k}\equiv 9\pmod{80}$. Then if $k\equiv 1\pmod{4}$ then 
\begin{eqnarray}
 \label{eq:4var_rep_MM1}
\nonumber\Tr_k(\alpha F(x))&=&\Tr_n\left(\alpha x^{2\cdot 3^k+1}\big(ax+a^5 x^{3^{2k}}\big)\right)\\
\nonumber&=&\Tr_n\Big((x_3 a^9+x_2 a^6+x_1 a^3+x_0)^2 (x_3 a^3+x_2 a^2+x_1 a+x_0)\\
\nonumber&&\ \times\big(x_3 a^4+x_2 a^3+x_1 a^2+x_0 a+(x_3 a^{32}+x_2 a^{23}+x_1 a^{14}+x_0 a^5)\big)\Big)\\
&=&\Tr_k\big(-\alpha x_1(x_0^3+x_0 x_2^2+x_2^3)+x_3\alpha(x_0^3-x_0^2 x_2-x_2^3)\big)\enspace.
\end{eqnarray}
Similarly, if $k\equiv 3\pmod{4}$ then
\begin{eqnarray}
 \label{eq:4var_rep_MM3}
\nonumber\Tr_k(\alpha F(x))&=&\Tr_n\left(\alpha x^{2\cdot 3^k+1}\big(a^5 x+ax^{3^{2k}}\big)\right)\\
\nonumber&=&\Tr_n\Big((x_3 a+x_2 a^{54}+x_1 a^{27}+x_0)^2 (x_3 a^3+x_2 a^2+x_1 a+x_0)\\
\nonumber&&\ \times\big(x_3 a^8+x_2 a^7+x_1 a^6+x_0 a^5+(x_3 a^{28}+x_2 a^{19}+x_1 a^{10}+x_0 a)\big)\Big)\\
&=&\Tr_k\big(\alpha x_0(-x_1^3+x_1^2 x_3+x_3^3)+\alpha x_2(x_1^2 x_3+x_1 x_3^2-x_3^3)\big)\enspace.
\end{eqnarray}
The obtained polynomial over $\fthreek$ is exceptional since its formula does not depend on the field extension degree.\qed 
\end{proof}

\begin{theorem}
 \label{th:binom3}
Let $n=4k$ with $k\equiv 1\pmod{4}$. Further, select $a\in\F_{3^4}$ such that $a^4-a^2-1=0$ and let $a_1=a^5$ and $a_2=a^{-5}$. Ternary 
vectorial function $F:\fthreen\mapsto\fthreek$ in (\ref{eq:binom}) is a quartic bent function having components in class $\cM$. 
\end{theorem}

\begin{proof}
Take any $\alpha\in\fthreek$ and using (\ref{eq:x_exp}), convert $f(x)$ in its univariate representation into the polynomial of four 
variables over $\fthreek$. Note that $3^k\equiv 3\pmod{80}$ for $k\equiv 1\pmod{4}$ and $3^{2k}\equiv 9\pmod{80}$. Then 
\begin{eqnarray}
 \label{eq:4var_rep_MM1}
\nonumber\Tr_k(\alpha F(x))&=&\Tr_n\left(\alpha x^{2\cdot 3^k+1}\big(ax+a^5 x^{3^{2k}}\big)\right)\\
\nonumber&=&\Tr_n\Big((x_3 a^9+x_2 a^6+x_1 a^3+x_0)^2(x_3 a^3+x_2 a^2+x_1 a+x_0)\\
\nonumber&&\ \times\big(x_3 a^8+x_2 c^7+x_1 a^6+x_0 a^5+(x_3 a^{22}+x_2 a^{13}+x_1 a^4+x_0 a^{11})\big)\Big)\\
&=&\Tr_k\big(\alpha x_1(x_0^3-x_0^2 x_2-x_2^3)-\alpha x_3(x_0^2 x_2+x_0 x_2^2-x_2^3)\big)\enspace.
\end{eqnarray}
The obtained polynomial over $\fthreek$ is exceptional since its formula does not depend on the field extension degree.\qed 
\end{proof}

\bibliography{Bent MM}

\end{document}